\documentclass[12pt]{article}
\usepackage{cite}
\usepackage{amsmath,amssymb,amsfonts,mathtools}
\usepackage{algorithmic}
\usepackage{graphicx}
\usepackage{textcomp}
\usepackage{xcolor}
\usepackage{amsthm}
\usepackage{mathabx}
\usepackage{hyperref}
\usepackage{caption}
\usepackage{tablefootnote}
\usepackage{array}

\usepackage{pgf,tikz,pgfplots}
\usepackage{mathrsfs}
\usetikzlibrary{arrows}

\usepackage{xpatch}
\makeatletter
\AtBeginDocument{\xpatchcmd{\@thm}{\thm@headpunct{.}}{\thm@headpunct{}}{}{}}
\makeatother
\usepackage{comment}

\DeclareCaptionType{equ}[][]
\newtheorem{theorem}{Theorem}[section]
\newtheorem{lemma}[theorem]{Lemma}
\newtheorem{prop}[theorem]{Proposition}
\newtheorem{cor}[theorem]{Corollary}
\newtheorem{example}[theorem]{Example}
\newtheorem{question}[theorem]{Question}
\newtheorem{remark}[theorem]{Remark}

\newcommand{\ben}{\begin{equation*}}
\newcommand{\een}{\end{equation*}}

\newcommand{\F}{\mathbb{F}}

\newcommand{\PP}{\mathbb{P}}
\newcommand{\Spa}{\operatorname{Span}}

\newcommand{\Hull}{\operatorname{Hull}}
\newcommand{\rank}{\operatorname{rank}}
\newcommand{\ev}{\operatorname{ev}}

\newcommand{\GL}{\operatorname{GL}}

\newcommand{\charac}{\operatorname{char}}

\begin{document}

\title{Hulls of Projective Reed-Muller Codes}
	
\author{Nathan Kaplan \\ Department of Mathematics \\ University of California, Irvine, USA \\
		{\tt nckaplan@math.uci.edu} \\ \\
Jon-Lark Kim \\ Department of Mathematics \\ Sogang University, Seoul, Korea \\
		{\tt jlkim@sogang.ac.kr} \\
	}
\date{\today}
	
\maketitle

\begin{abstract}
Projective Reed-Muller codes are constructed from the family of projective hypersurfaces of a fixed degree over a finite field $\F_q$.  We consider the relationship between projective Reed-Muller codes and their duals.  We determine when these codes are self-dual, when they are self-orthogonal, and when they are LCD.  We then show that when $q$ is sufficiently large, the
dimension of the hull of a projective Reed-Muller code is 1 less than the dimension of the code. We determine the dimension of the hull for a wider range of parameters and describe how this leads to a new proof of a recent result of Ruano and San Jos\'e (2024).
\end{abstract}



\section{Introduction}

Affine Reed-Muller codes are the multivariate analogues of Reed-Solomon codes. Codewords come from evaluating polynomials in $n$ variables at the points of $\F_q^n$, or equivalently, they correspond to affine hypersurfaces.  There is an extensive literature studying these codes, including foundational results of Delsarte, Goethals, and MacWilliams \cite{DelGoeMac}, Kasami, Lin, and Peterson \cite{KasLinPet}, and many others.

This paper deals with projective analogues of these codes, the projective Reed-Muller codes, where codewords correspond to hypersurfaces of a fixed degree $k$ in $\PP^n$ defined over $\F_q$.  These codes also have an extensive history going back to early work of Lachaud who studied codes from hyperplanes and codes from quadric hypersurfaces in \cite{Lac1}, and determined the parameters of these codes in the case where $k < q$ in \cite{Lac2}.  S\o{}rensen determined the parameters of these codes in general, described the dual of a projective Reed-Muller code, and studied when these codes are cyclic \cite{Sor}.  Berger analyzes the automorphism groups of these codes in \cite{Ber}.  More recently, there has been extensive interest in determining the generalized Hamming weights of these codes, and in determining Hamming weight enumerators of these codes in certain cases; see for example \cite{BeeDatGho2, BeeDatGho, Elk, Kap1}.

The main focus of this paper is on studying the hulls of projective Reed-Muller codes.  The \emph{hull of a linear code} $C$ is defined by $\Hull(C) = C \cap C^\perp$.  Three particular cases of interest are when $C$ is \emph{self-dual}, which means $C = C^\perp = \Hull(C)$, when $C$ is \emph{self-orthogonal}, which means $C = \Hull(C)$, and when $C$ is \emph{LCD (linear code with complementary dual)}, which means $\Hull(C)$ is trivial.  In addition to the huge literature of self-dual codes, self-orthogonal codes, and LCD codes (see for example \cite{Kim1, Roe1, Massey1, Massey2}), there have been many recent papers studying hulls of families of codes more generally, some of which is motivated by connections to quantum error-correcting codes.

The CSS construction gives a method for constructing a quantum error-correcting code from a classical linear code, but this construction requires that the classical code be self-orthogonal.  Mathematicians have worked around this self-orthogonality restriction by using \emph{entanglement assistance}.  This technique allows for the construction of a quantum code from a pair of classical codes $C_1, C_2 \subseteq \F_q^n$, and the dimension of the \emph{relative hull of $C_1$ with respect to $C_2$}, $C_1\cap C_2^\perp$, plays an important role in analyzing the parameters.

Ruano and San-Jos\'e have recently determined the dimension of the hulls of projective Reed-Muller codes over the projective plane \cite{RSJ1}.  In fact, they study intersections of pairs of projective Reed-Muller codes over the projective plane more generally.  We will review these contributions in Section \ref{sec:RSJ}.  We give a different proof of one of their results in Section \ref{sec:main2}.  The results of our paper apply not only to projective Reed-Muller codes over the projective plane, but also to projective Reed-Muller codes over higher-dimensional projective spaces.  We also note recent work of Gao, Yue, Huang, and Zhang \cite{GaoYueHuaZha}, and of Chen, Ling, and Liu \cite{CheLinLiu}, on hulls of generalized Reed-Solomon codes.

Ruano and San-Jos\'e investigate applications of their results to the construction of quantum codes with good parameters in \cite{RSJ1}.  They study the construction of quantum codes from projective Reed-Muller codes further in \cite{RSJ2}, which focuses on what they call the \emph{hull variation problem for projective Reed-Muller codes}.  Instead of restricting themselves to the study of $C \cap C^\perp$ where $C$ is a projective Reed-Muller code, or even to $C_1 \cap C_2^\perp$ where $C_1$ and $C_2$ are projective Reed-Muller codes, they consider $C_{1,\ell} \cap C_2^\perp$ where $C_2$ is a projective Reed-Muller code and $C_{1,\ell}$ is a code that is monomially equivalent to a projective Reed-Muller code.  By varying the first code one can vary the dimension of the intersection, which is sometimes desirable for applications.  The idea of varying the intersection by replacing one code with a code monomially equivalent to it also occurs in work of Anderson et al. \cite[Theorem 4.6]{ACMLMRS}, and in work of Guenda et al. \cite[Theorem 2.2]{GGJT}.  We do not consider these hull variation problems or the application to quantum codes in this paper and leave this as a direction for future work.

The main contributions of our paper are as follows.  Building on S\o{}rensen's analysis of the dual of a projective Reed-Muller code, we completely determine when projective Reed-Muller codes are self-dual, when they are self-orthogonal, and when they are LCD.  We then show that when $q$ is sufficiently large, the dimension of the hull of a projective Reed-Muller code is $1$ less than the dimension of the code.  We then determine the dimension of the hull for a wider range of parameters and explain how this leads to a new proof of a recent result of Ruano and San Jos\'e \cite{RSJ1}. We give examples throughout the paper to illustrate our results.

Our paper consists of four sections. Section \ref{sec:prelim} describes preliminaries on projective Reed-Muller codes, the hull of a linear code, and duals of projective Reed-Muller codes. We prove our main results in Sections \ref{sec:main1} and \ref{sec:main2}.  We end the paper by highlighting parameters for which our results do not determine the dimension of the hull, which suggests a natural direction for future work.

\section{Preliminaries}\label{sec:prelim}
For a general reference on coding theory we recommend \cite{Pless1, Macwilliams, Pless2}.

\subsection{Projective Reed-Muller Codes}

 Let $\mathcal P=\{P_1, \dots , P_N \}$ be a subset of points of the affine space $\F_q^n$
and let $\mathcal V$ be a finite subspace of polynomials in $\F_q[x_1, \dots, x_n]$. Consider the evaluation map:
\begin{eqnarray*}
  {\mbox{ev}}_{\mathcal P}: \mathcal V  & \rightarrow & \F_q^N \\
   f & \rightarrow & (f(P_1), \dots, f(P_N)).
\end{eqnarray*}
The {\em evaluation code} $\ev(\mathcal{V},P)$ is the image of this map. The {\em affine Reed-Muller code} of order $k$ and length $q^n$, which we denote by ${\mbox{RM}}_q^A
(k, n)$, comes from choosing $\mathcal V$ to be
$\F_q[x_1, \dots, x_n]_{\le k}$, the vector space of polynomials of degree at most $k$, and $\mathcal P$ to be
the set of all points in the affine space $\F_q^n$. This is sometimes referred to as a \emph{generalized Reed-Muller code} in the literature.  Note that ${\mbox{RM}}_q^A(k, n)$ depends on a choice of the ordering of the points of $\F_q^n$.  Recall that two linear codes over $\F_q$ are \emph{monomially equivalent}~\cite{Pless1} if one code can be obtained from the other by permuting the coordinate positions and multiplying each coordinate by a nonzero element of $\F_q$. It is clear that the affine Reed-Muller codes that arise from different choices are all monomially equivalent.

Let $\mathcal P =\{ P_1, \dots , P_N \}$ be a subset of points in the projective space $\mathbb P^n (\F_q)$.
It does not make sense to evaluate an element of $\F_q[x_0, \dots, x_n]$ at a projective point, so we make a choice of affine representative for each, giving a set $\mathcal P'=\{P_1', \dots, P_N'\}$ with each $P_i' \in \F_q^{n+1} \setminus (0,0,\ldots,0)$.
 Let $\mathcal V$ be a subspace of $\F_q[x_0, \dots, x_n]_k$, the
set of homogeneous polynomials of degree $k$ (including the zero polynomial). We
now define an evaluation code as in the previous paragraph. When $\mathcal P$ consists
of all points in $\mathbb P^n(\F_q)$ and $\mathcal V$ is all of $\F_q[x_0, \dots, x_n]_k$, this construction defines the \emph{projective Reed-Muller code of order $k$ and length} $N = |\mathbb P^n(\F_q)| = (q^{n+1}-1)/(q-1),$
denoted ${\mbox{RM}}_q^P(k, n)$.
In this paper, we write $C_{n,k}^q$ for ${\mbox{RM}}_q^P(k, n)$
and $C_{n,k}^{A,q}$ for ${\mbox{RM}}_q^A(k, n)$. As in the previous paragraph, the projective Reed-Muller codes that arise from different orderings of the points of $\mathcal{P}$ and different choices of affine representatives in $\mathcal{P}'$ are monomially equivalent.  In this paper we follow the convention of Lachaud \cite{Lac1, Lac2} and S\o{}rensen \cite{Sor}, where for each projective point we choose the affine representative for which the left-most nonzero coordinate is equal to $1$.  Throughout the rest of this paper, whenever we write $C_{n,k}^q$ we will use this choice of affine representatives.

To give one basic example, when $k=1$ the projective Reed-Muller code arises from evaluating linear forms on projective space and one can see that $C_{n,1}^q$ is a $q$-ary simplex code.  Lachaud discusses this example, and $C_{n,2}^q$ where codewords come from evaluating quadratic polynomials, in his initial paper on this topic \cite{Lac1}.

Lachaud determined the parameters of the code $C_{n,k}^q$ in the case where $1 \le k < q$.
\begin{lemma}{\rm (\cite[Theorem 2]{Lac2})} \label{lem-parameters}
Assume that $1 \le k < q$. Then the code $C_{n,k}^q$ has parameters
{\rm length} $N=\frac{q^{n+1}-1}{q-1}$, {\rm dimension} $K=\binom{n+k}{k}$,
{\rm distance} $D=(q-k+1)q^{n-1}$.
\end{lemma}

If $k \ge n(q-1) +1$, then $C_{n,k}^q$ is trivial, that is, the whole space $\mathbb F_q^N$~\cite[Remark 3]{Sor}.
For $1 \le k \le n(q-1)$, S\o{}rensen computed the parameters of $C_{n,k}^q$.  For a recent discussion of S\o{}rensen's computation of the minimum distance and also for a characterization of the minimal weight codewords of $C_{n,k}^q$, see the paper of Ghorpade and Ludhani \cite{GhoLud}, and also the short note of S\o{}rensen \cite{Sor2}.

\begin{lemma}{\rm(\cite[Theorem 1]{Sor})} \label{lem-parameters-2}
Assume that $1 \le k \le n(q-1)$. Then the code $C_{n,k}^q$ has parameters
\begin{eqnarray*}
{\mbox{{\rm length} }} N & = &\frac{q^{n+1}-1}{q-1}, \\
{\mbox{{\rm dimension} }} K &= &\sum_{\substack{t \equiv k\hspace{-.22cm} \pmod{q-1} \\ 0 < t \le k}} \left(
\sum_{j=0}^{n+1} (-1)^j \binom{n+1}{j} \binom{t-jq+n}{t-jq} \right), \\
{\mbox{{\rm distance} }} D & = & (q-s)q^{n-r-1},
\end{eqnarray*}
where
\[
k-1 = r(q-1) +s,~~ 0 \le s < q-1.
\]
\end{lemma}

We remark that if $k <q$, then $r=0$ and $s=k-1$ so that
\[
D=(q-s)q^{n-r-1}=(q-(k-1))q^{n-1}=(q-k+1)q^{n-1},
\]
which matches the expression for $D$ in Lemma~\ref{lem-parameters}.

A different expression for the dimension of $C_{n,k}^q$ is given by Mercier and Rolland~\cite{MerRol}. See \cite{GhoLud} for a discussion of the equivalence between these two expressions.
\begin{lemma}{\rm (\cite[Theorem 4]{MerRol})} \label{lem-dim-MR}
Assume that $1 \le k \le n(q-1)$. Then $C_{n,k}^q$ has dimension
\[
K = \binom{n+k}{k} - \sum_{j=2}^{n+1} (-1)^j \binom{n+1}{j} \sum_{i=0}^{j-2} \binom{k+(i+1)(q-1)-jq+n}{k+(i+1)(q-1)-jq}.
\]
\end{lemma}

We briefly recall some notation related to polynomials over finite fields.  Let $F(x_0,\ldots, x_n) \in \F_q[x_0,\ldots, x_n]$.  The \emph{reduced form of $F(x_0,\ldots, x_n)$}, denoted by $\overline{F}(x_0,\ldots, x_n)$, is the polynomial we get by replacing any factor $x_j^{t_j}$, where $t_j = a(q-1) + b$ with $0 < b \le q-1$, with $x_j^b$.  That is, $\overline{F}(x_0,\ldots, x_n)$ is the polynomial we get by reducing modulo $x_j^q-x_j$ for each $j$.
\begin{prop}\label{prop-dim_increase}
For $k$ satisfying $1\le k < n(q-1)$, we have $\dim(C_{n,k}^q) < \dim(C_{n,k+1}^q)$.
\end{prop}
We could give a proof of this using the expression for $\dim(C_{n,k}^q)$ given in Lemma \ref{lem-parameters-2}, or the one given in Lemma \ref{lem-dim-MR}, but because these expressions are complicated in general we give a more combinatorial argument.

\begin{proof}
In the proof of \cite[Theorem 2]{Sor}, S\o{}rensen explains that the dimension of $C_{n,k}^q$ is equal to the number of distinct elements $\overline{m}$ where $m$ is a monomial in $x_0,\ldots, x_n$ of degree $k$.  We call a monomial $m$ \emph{reduced} if all of its exponents are at most $q-1$.  This is equivalent to $\overline{m} = m$.  We see that $\dim(C_{n,k}^q)$ is equal to the number of reduced monomials in $x_0,\ldots, x_n$ of degree $k - j(q-1)$ where $j$ can take any value satisfying $0 \le j < \frac{k}{q-1}$.

Let $\mathcal{R}_{n,k}^q$ denote the set of reduced monomials $\overline{m}$ as $m$ runs through all monomials of degree $k$ in $x_0,\ldots, x_n$.  We will give an injective function $\varphi\colon \mathcal{R}_{n,k}^q \rightarrow \mathcal{R}_{n,k+1}^q$ and then will prove that there exists an element of $\mathcal{R}_{n,k+1}^q\setminus \varphi\left(\mathcal{R}_{n,k}^q\right)$.  This shows that $|\mathcal{R}_{n,k+1}^q| > |\mathcal{R}_{n,k}^q|$, completing the proof.

Define $\varphi\colon \mathcal{R}_{n,k}^q \rightarrow \mathcal{R}_{n,k+1}^q$ by $\varphi(m) = \overline{x_0  m}$.  There are two cases to consider.
\begin{enumerate}
\item First suppose that $x_0  m = \overline{x_0  m}$.  This occurs precisely when the exponent of $x_0$ in $m$ is less than $q-1$.  If $m = \overline{F}$ for some homogeneous polynomial $F$ of degree $k$, then it is clear that $x_0 m = \overline{x_0  F}$.  Therefore, $x_0  m \in \mathcal{R}_{n,k+1}^q$.

\item Now suppose that $x_0  m \neq \overline{x_0 m}$.  This occurs precisely when the exponent of $x_0$ in $m$ is $q-1$.  In this case, $\overline{x_0 m} = \frac{x_0 m}{x_0^{q-1}}$ and $\frac{x_0 m}{x_0^{q-1}}$ is reduced.  As in the previous case, if $m = \overline{F}$ for  some homogeneous polynomial $F$ of degree $k$, then it is clear that $\overline{x_0 m} = \overline{x_0  F}$.  Therefore, $\overline{x_0  m} \in \mathcal{R}_{n,k+1}^q$.
\end{enumerate}
We have shown that $\varphi\colon \mathcal{R}_{n,k}^q \rightarrow \mathcal{R}_{n,k+1}^q$ is injective.  We now show that there exists an element of $\mathcal{R}_{n,k+1}^q\setminus \varphi\left(\mathcal{R}_{n,k}^q\right)$.

The definition of $\varphi$ implies that the exponent of $x_0$ is nonzero in every monomial in $\varphi\left(\mathcal{R}_{n,k}^q\right)$.  Since $k+1 \le n(q-1)$, there is at least one monomial $x_0^0x_1^{a_1}\cdots x_n^{a_n}$ for which $\sum_{i=1}^n a_i = k+1$ and $0 \le a_i \le q-1$ for each $i$.  That is, there is at least one reduced monomial of degree $k+1$ in which the exponent of $x_0$ is $0$.  This is an element of $\mathcal{R}_{n,k+1}^q \setminus \varphi\left(\mathcal{R}_{n,k}^q\right)$.
\end{proof}

The following fact follows immediately.  We will apply it in our proof of Theorem \ref{thm-so-1}.
\begin{cor}\label{cor-dim_diff}
Suppose $a$ and $b$ are integers satisfying $1 \le a < b \le n(q-1)$.  Then $\dim(C_{n,b}^q) - \dim(C_{n,a}^q) \ge b-a$.
\end{cor}

\subsection{The Hull of a Linear Code}
Suppose $x = (x_1,\ldots, x_N)$ and $y = (y_1,\ldots, y_N)$ are elements of $\F_q^N$. Let $\langle x,y \rangle = \sum_{i=1}^N x_i y_i$.  The \emph{dual} of a linear code $C \subseteq \F_q^N$ is
\[
C^\perp = \{y \in \F_q^N \colon \langle y,c\rangle = 0\ \text{ for all } c\in C\}.
\]
The \emph{hull} of a linear code $C$ is defined by $\Hull(C) = C\cap C^\perp$.  A linear code is \emph{self-dual} if $C = C^\perp$ and is \emph{self-orthogonal} if $C \subseteq C^\perp$.  A $K \times N$ matrix $G$ over $\F_q$ is a \emph{generator matrix} for a linear code $C\subset\F_q^N$ of dimension $K$ if the rows of $G$ form a basis for $C$.  We recall a basic fact about the dimension of the hull of a linear code $C$.
\begin{prop}{\rm(\cite[Proposition 3.1]{GueJitGul})} \label{prop-hull}
Let $C \subseteq \F_q^N$ be a linear code of dimension $K$ and let $G$ be a generator matrix for $C$.  Then $\rank(G G^T) = K - \dim(\Hull(C))$.

In particular, $C$ is self-orthogonal if and only if $GG^T$ is the zero matrix and $C$ is LCD if and only if $GG^T$ is invertible.
\end{prop}

There is a result like this one that can be used to determine the dimension of $C_1 \cap C_2^\perp$ when $C_1, C_2 \subseteq \F_q^N$ are linear codes that are not necessarily equal \cite[Theorem 2.1]{GGJT}.  We will not apply this result in this paper, but at the end of Section \ref{sec:main2} when we discuss further directions for our work we will return to the role that this theorem would play in our strategy.

\subsection{Duals of Projective Reed-Muller Codes}

In order to study hulls of projective Reed-Muller codes, we need to understand the dual of a projective Reed-Muller code.  It is well known that the dual of an affine Reed-Muller code is also an affine Reed-Muller code.  More precisely, $C_{n,k}^{A,q}\ ^\perp = C_{n,n(q-1)-k-1}^{A,q}$; see for example \cite{DelGoeMac}.  Since $\F_q[x_1, \dots, x_n]_{\le k} \subseteq \F_q[x_1, \dots, x_n]_{\le k+1}$, it is clear that if $\ell = n(q-1)-k-1$ and $k' = \min(k,\ell)$ then $C_{n,k'}^{q,A}$ is self-orthogonal.  Therefore, in all cases we can determine~$\Hull(C_{n,k}^{A,q})$.

The projective analogue of this situation is more complicated since it is not always the case that the dual of a projective Reed-Muller code is a projective Reed-Muller code. S\o{}rensen computes the dual of a projective Reed-Muller code in \cite{Sor}.  Throughout this paper, we write  ${\bf 1}=(1,1, \dots, 1)$.  We recall the convention that $C_{n,0}^q =\mbox{Span}_{\F_q}  \{ {\bf 1}\}$, since the homogeneous polynomials of degree $0$ in $x_0,\ldots, x_n$ are the constant functions.
\begin{theorem}{\rm (\cite[Theorem 2]{Sor})} \label{thm-Sor} Let $k$ be an integer satisfying $1 \le k \le n(q-1)$ and let $\ell =n(q-1)-k$. Then
\begin{enumerate}
\item [{(i)}] ${C_{n,k}^q}^{\perp} = C_{n,\ell}^q$ for $k \not \equiv 0 \pmod{q-1}$.
\item [{(ii)}] ${C_{n,k}^q}^{\perp} = {{\Spa}}_{\F_q}  \{ {\bf 1}, C_{n,\ell}^q \}$ for $k  \equiv 0 \pmod{q-1}$.
\end{enumerate}
\end{theorem}

\begin{remark}\label{rem-no1}{\rm 
Suppose $k$ is an integer satisfying $1 \le k < n(q-1)$ and $k \equiv 0 \pmod{q-1}$.  Let $\ell = n(q-1)-k$.  Theorem \ref{thm-Sor} implies that ${C_{n,k}^q}^\perp = {\mbox{Span}}_{\F_q}  \{ {\bf 1}, C_{n,\ell}^q \}$.  In these cases ${\bf 1} \not\in C_{n,\ell}^q$, so we have $C_{n,\ell}^q \subsetneq {C_{n,k}^q}^{\perp} $.  S\o{}rensen does not explain why ${\bf 1} \not\in C_{n,\ell}^q$, so we give a quick proof of this fact.

Suppose ${\bf 1} \in C_{n,\ell}^q$.  Note that $\ell \equiv -k \equiv 0 \pmod{q-1}$ and that $\ell > 0$. At the end of the proof of Theorem \ref{thm-Sor}, S\o{}rensen shows that ${\bf 1} \in {C_{n,\ell}^q}^\perp$.  If $m$ is a positive integer we have $m \cdot 1 = 0$ in $\F_q$ if and only if the characteristic of $\F_q$ divides $m$.  However,
\[
\langle {\bf 1}, {\bf 1}\rangle =\frac{q^{n+1}-1}{q-1} \cdot 1 = (q^n + q^{n-1}+\cdots + q + 1)\cdot 1.
\]
Since $q$ is a power of $\charac(\F_q)$ it is clear that $q^n + q^{n-1}+\cdots + q + 1\equiv 1 \pmod{\charac(\F_q)}$.  This implies $\langle {\bf 1}, {\bf 1}\rangle \neq 0$ in $\F_q$, contradicting the fact that ${\bf 1} \in {C_{n,\ell}^q}^\perp$.  

Now suppose $k = n(q-1)$.  Since ${C_{n,n(q-1)}^q}^\perp = {\mbox{Span}}_{\F_q}  \{ {\bf 1}\}$, the argument in the previous paragraph shows that ${\bf 1} \not\in C_{n,n(q-1)}^q$. This observation implies that $\Hull(C_{n,n(q-1)}^q)$ is trivial, or equivalently, that $C_{n,n(q-1)}^q$ is an LCD code. }
\end{remark}

\subsection{Hulls of Projective Reed-Muller Codes over the Projective Plane}\label{sec:RSJ}

Ruano and San-Jos\'e study projective Reed-Muller codes over the projective plane in \cite{RSJ1}.  That is, they focus on the $n=2$ case of our general study of the codes $C_{n,k}^q$.  Ruano and San-Jos\'e give a basis for $\Hull(C_{2,k}^q)$ for any $1 \le k \le q-1$ in the following form.
\begin{theorem}\cite[Corollary 3.12]{RSJ1}\label{RSJ312}
Let $1\le k \le q-1$ and let $Y = \{0,1,\ldots, \min\{k-1,q-1-k-1\}\}$.
\begin{itemize}
\item If $2k \equiv 0 \pmod{q-1}$, then $C_{2,k}^q \cap {C_{2,k}^q}^\perp = C_{2,k}^q$.
\item If $2k \not\equiv 0 \pmod{q-1}$, a basis for $C_{2,k}^q \cap {C_{2,k}^q}^\perp$ is given by
\[
A_1^k \cup \left( \bigcup_{a_2 \in Y} x_1^{k-a_2} x_2^{a_2} \right),
\]
where
\[
A_1^k = \{x_0^{a_0} x_1^{a_1} x_2^{a_2} \colon a_0 > 0,\ a_0 + a_1 + a_2 = k,\ 0 \le a_1,a_2 \le q-1\}.
\]
\end{itemize}
Consequently,
\[
\dim(\Hull(C_{2,k}^q)) = \binom{k+1}{2} + \min\{k,q-k-1\}.
\]
\end{theorem}
For any $k$ satisfying $1 \le k < 2(q-1)$ except $k = q-1$, Theorem \ref{thm-Sor}(i) implies that ${C_{2,k}^q}^\perp = C_{2,2(q-1)-k}^q$. If $k$ satisfies $q-1 < k < 2(q-1)$, then ${C_{2,k}^q}^\perp = C_{2,\ell}^q$ where $\ell = 2(q-1)-k$ satisfies $1 \le \ell < q-1$, and $\Hull(C_{2,k}^q) = \Hull(C_{2,\ell}^q)$.  In this way, we see that Theorem \ref{RSJ312} gives a basis for $\Hull(C_{2,k}^q)$ for all $k$ satisfying $1 \le k \le 2(q-1)$.

Theorem \ref{RSJ312} follows from the results we prove in Section \ref{sec:main2}.  We emphasize that our proof is quite different than the one given by Ruano and San-Jos\'e.  A main idea in \cite{RSJ1} is to view codewords of projective Reed-Muller codes as classes of polynomials in a quotient ring.  They then compute bases of some subspaces of quotient rings associated to projective Reed-Muller codes by using techniques from commutative algebra.  In particular, the authors work with a particular basis of the vanishing ideal of the collection of points of $\PP^2(\F_q)$.  Our proofs work directly with the linear algebraic setup of Proposition \ref{prop-hull}, taking sums of polynomials evaluated at the standard affine representatives of the points of $\PP^n(\F_q)$ to determine the rank of the relevant matrix.

We note that Ruano and San-Jos\'e actually consider the more general problem of determining a basis for $C_{2,k_1}^q \cap C_{2,k_2}^q$ where $1 \le k_1, k_2, \le 2(q-1)$. Theorem \ref{RSJ312} comes from specializing to the case where $k_2 = 2(q-1)-k_1$ \cite[Corollary 3.11]{RSJ1}.  In Section \ref{sec:main2} we discuss how one could adapt our techniques to ammept to give a different proof of this more general statement about projective Reed-Muller codes with $n=2$.

\section{Special Classes of Projective Reed-Muller Codes}\label{sec:main1}

Our goal is to understand the following question.
\begin{question}\label{Q1}
\begin{enumerate}
\item For which parameters $(n,k,q)$ is the projective Reed-Muller code $C_{n,k}^q$ self-dual?

\item For which parameters $(n,k,q)$ is $C_{n,k}^q$ self-orthogonal?

\item For which parameters $(n,k,q)$ is $C_{n,k}^q$ LCD?

\item In cases where $C_{n,k}^q$ is not self-orthogonal, what can we say about $\Hull(C_{n,k}^q)$?
\end{enumerate}
\end{question}
\noindent In this section we answer the first three questions.  We consider the last one in Section \ref{sec:main2}.

We begin this section by applying Theorem \ref{thm-Sor} to answer the first part of Question \ref{Q1} and determine when a projective Reed-Muller code is self-dual.
\begin{theorem}\label{thm-proj-sd}
Let $1\le k \le n(q-1)$.  Then $C^q_{n,k}$ is self-dual if and only if $q$ and $n$ are odd and $k = \frac{n(q-1)}{2}$.
\end{theorem}

\begin{proof}
Suppose that $q$ and $n$ are odd and $k = \frac{n(q-1)}{2}$.  Then $k \not\equiv 0 \pmod{q-1}$ because $n$ is odd.  Note that $n(q-1)-k = k$. Theorem~\ref{thm-Sor}(i) implies that ${C_{n,k}^q}^{\perp} = {C_{n,k}^q}$, proving that ${C_{n,k}^q}$ is self-dual.

Conversely, suppose that $C^q_{n,k}$ is self-dual. Then its length $N=q^n + q^{n-1} + \cdots + 1$ is even. If $q$ is even, then $N$ is odd. Hence $q$ is odd.
We complete the proof by considering two cases.
\begin{itemize}
\item Suppose that $k \not\equiv 0 \pmod{q-1}$. Then by Theorem~\ref{thm-Sor}(i), we have
\[{C_{n,k}^q}^{\perp} = C_{n,\ell}^q=C_{n,n(q-1)-k}^q.
\]
Proposition \ref{prop-dim_increase} implies that for any $1 \le k,k' \le n(q-1)$ with $k \neq k'$ the codes $C_{n,k}^q$ and $C_{n,k'}^q$ have different dimensions.
\begin{comment}
Lemma \ref{lem-parameters-2} implies that for any $1 \le k,k'\le n(q-1)$ with $k \neq k'$ the codes $C_{n,k}^q$ and $C_{n,k'}^q$ have different minimal distances. The reason is as follows. Suppose $k - 1 = r(q-1) + s$ where $0 \le s < q-1$ and $k'-1 = r'(q-1)+s'$ where $0 \le s' < q-1$.  Without loss of generality, suppose $k' \ge k$, so that $r' \ge r$.  We see that $D=(q-s)q^{n-r-1} = (q-s')q^{n-r'-1}=D'$ if and only if $q^{r'-r} = \frac{q-s'}{q-s}$.  Since $2\le q-s, q-s' \le q$, the only way for $\frac{q-s'}{q-s}$ to be a nonnegative power of $q$ is for it to be equal to $1$.  This implies $s'=s$.  Then, $q^{r'-r} = 1$ implies $r=r'$.  Since $r=r'$ and $s=s'$, we conclude that $k=k'$. This is a contradiction. Hence, $C_{n,k}^q$ and $C_{n,k'}^q$ have different minimal distances.
\end{comment}
Therefore, ${C_{n,k}^q}^{\perp}=C_{n,\ell}^q = C_{n,k}^q$ implies $\ell = n(q-1)-k=k$, that is, $k =  \frac{n(q-1)}{2}$.
If $n$ is even, then $k = \frac{n(q-1)}{2} \equiv 0 \pmod{q-1}$, which is a contradiction. Hence, $n$ is odd.

\item Suppose that $k \equiv 0 \pmod{q-1}$. Then we have
 ${C_{n,k}^q}^{\perp} = {\mbox{Span}}_{\F_q}  \{ {\bf 1}, C_{n,n(q-1)-k}^q \}$ by Theorem~\ref{thm-Sor}(ii).
 Since $C^q_{n,k}$ is self-dual, we know that
 $$C^q_{n,k}={C_{n,k}^q}^{\perp} = {\mbox{Span}}_{\F_q}  \{ {\bf 1}, C_{n,n(q-1)-k}^q \}.$$ This implies that ${\bf 1} \in C^q_{n,k}$. By the remark following Theorem \ref{thm-Sor}, this is not possible. Therefore, $C^q_{n,k}$ is not self-dual.
\end{itemize}
\end{proof}

\begin{example}{\em
Let $q=3$. In order to apply Theorem~\ref{thm-proj-sd}, we must have $n$ odd and $k = n$.

For each odd $n$, $C_{n,n}^3$ is a self-dual code with parameters
$[N=\frac{3^{n+1} -1}{2}, K=N/2, D=3^{\frac{n+1}{2}}]$. In particular, $C_{1,1}^3$ is a ternary self-dual $[4, 2, 3]$ code, which is known as the tetracode. When $k=3$, $C_{3,3}^3$ is a ternary self-dual $[40, 20, 9]$ code.
Elkies~\cite{Elk} showed that this code is self-dual in his work on the weight enumerator of the projective Reed-Muller code $C_{3,3}^q$.  The weight enumerator of this code is given by
\begin{eqnarray*}
& & x^{40} + 1040 x^{31} y^9 +  18720 x^{28} y^{12} +  1100736 x^{25} y^{15} + 25761840 x^{22}y^{18} \\
&+ & 236377440 x^{19} y^{21} + 908079120 x^{16}y^{24} + 1388750720 x^{13}y^{27} \\
& +& 783679104 x^{10} y^{30} + 137535840 x^7 y^{33} + 5468320 x^4y^{36} + 11520 xy^{39}.
\end{eqnarray*}
A computation in Magma shows that the automorphism group of $C_{3,3}^3$ is isomorphic to $\GL_4(\F_3)$.  We refer to~\cite{Ber} for a detailed discussion of automorphism groups of projective Reed-Muller codes. It is also interesting to note that the set of codewords of weight 9 form a 2-design. More precisely, there are 1040 codewords of weight 9. Since we are only considering the supports of these vectors over $\mathbb{F}_3$, there are exactly 520 distinct supports (also called blocks). We have checked using Magma that for any two distinct positions out of 40, there are exactly 24 blocks containing them, hence we get $\lambda=24$. Therefore we have a 2-$(40, 9, 24)$ design.
}

\end{example}

\begin{example}{\em
Let $q=5$. Theorem~\ref{thm-proj-sd} implies that $C_{n,k}^5$ is self-dual if and only if $n$ is odd and $k=2n$. Note that $k-1=2n-1=2(2n'+1)-1 = 4n'+1$ for some integer $n' \ge 0$. If we write $k-1 = 4r + s$ where $0 < s \le 4$, we find $r = n'$ and $s=1$.  This implies $r=\frac{2n-2}{4} =\frac{n-1}{2}$.
The code $C_{n,k}^5$ has parameters
$[N=\frac{5^{n+1} -1}{4}, K=N/2, D=4 \cdot 5^{\frac{n-1}{2}}]$.

 In particular, $C_{1,2}^5$ is a self-dual code with parameters $[6, 3, 4]$ over $\F_5$. It is also an MDS code. This code is unique up to $(1, -1)$-monomial matrix equivalence and is known as $F_6$ in the notation of~\cite{LeoPleSlo}.  A generator matrix of this code is given by
\[
\left[
\begin{array}{cccccc}
1 & 0 & 1 & 4 & 4 & 1 \\
1 & 1 & 0 & 1 & 4 & 4  \\
1 & 4 & 1 & 0 & 1 &4 \\
\end{array}
\right].
\]
}
\end{example}

In several places throughout the rest of the paper, we construct a choice of generator matrix for $C_{n,k}^q$.  This involves the evaluation map $\ev(f)$, which depends on an ordered choice of affine representatives for the points of $\PP^n(\F_q)$.  Recall that we are following the convention that for each projective point, we choose the affine representative for which the left-most nonzero coordinate is equal to $1$.  We fix an arbitrary ordering of these affine representatives to get $\mathcal{P}'$.  We will use this ordered choice of affine representatives $\mathcal{P}'$ throughout the rest of the paper.

We now answer the second the second part of Question \ref{Q1} and determine the projective Reed-Muller codes that are self-orthogonal.
\begin{theorem}\label{thm-so-1}
Suppose that $1\le k \le n(q-1)$.  Then $C_{n,k}^q$ is self-orthogonal if and only if $1 \le k \le \frac{n(q-1)}{2}$ and $2k \equiv 0 \pmod{q-1}$.
\end{theorem}

\begin{proof}
Suppose that $1 \le k \le \frac{n(q-1)}{2}$ and $2k \equiv 0 \pmod{q-1}$. Since $2k \equiv 0 \pmod{q-1}$ and $k \le \frac{n(q-1)}{2}$, we have $\ell =n(q-1)-k = k + r (q-1)$ for some $r \ge 0$.
 By Theorem \ref{thm-Sor}, $C_{n,n(q-1)-k}^q \subseteq {C_{n,k}^q}^\perp$.

 In order to prove $C_{n,k} \subseteq {C_{n,k}^q}^\perp$, it suffices to show that $C_{n,k} \subseteq C_{n,n(q-1)-k}^q$. For each monomial $m$ of degree $k$, we construct a monomial $m'$ of degree $n(q-1)-k$ such that $\overline{m'} = m$.  Let $x_0^{a_1}\cdots x_n^{a_n}$ be a monomial of degree $k$, so $0 \le a_0,\ldots, a_n \le k$ and $\sum_{i=0}^n a_i = k$.  Let $j$ be the smallest positive integer such that $a_j >0$. The evaluation map sends a polynomial and the reduced form of that polynomial to the same element of $\F_q^N$.  More precisely, $x_j^{r(q-1)+a_j} = x_j^{a_j}$ at every element of $\F_q$. Therefore
\[
\ev_{\mathcal{P}}(x_0^{a_1}\cdots x_n^{a_n} x_j^{r(q-1)})=
\ev_{\mathcal{P}}(x_0^{a_1}\cdots x_n^{a_n}).
\]
We see that the span of $\ev_{\mathcal{P}}(x_0^{a_1}\cdots x_n^{a_n})$ as we vary over all monomials in $x_0,\ldots, x_n$ of degree $k$ is a subcode of $C_{n,n(q-1)-k}^q$, completing the proof that $C_{n,k}^q$ is self-orthogonal.

We now show that $C_{n,k}^q$ is not self-orthogonal in all other cases.  Suppose $\frac{n(q-1)}{2} < k \le n(q-1)$.  Let $\ell = n(q-1)-k$, so $0 \le  \ell < \frac{n(q-1)}{2}$ and $k - \ell = 2k - n(q-1)$.  We consider two cases.
\begin{enumerate}
\item Suppose $k - \ell = 1$.  This implies $k = \frac{n(q-1)+1}{2}$, so in particular $q$ is even and $n$ is odd.  Note that $2k = n(q-1) + 1 \equiv 1 \pmod{q-1}$.  If $q \neq 2$, then $k \not\equiv 0 \pmod{q-1}$ and Theorem \ref{thm-Sor}(i) implies that ${C_{n,k}^q}^\perp = C_{n,\ell}^q$.  Proposition \ref{prop-dim_increase} implies that $\dim({C_{n,k}^q}^\perp) < \dim(C_{n,k}^q)$, and therefore $C_{n,k}^q$ is not self-orthogonal.

Suppose $q=2$.  If $n=1$, then $k = 1$. It is easy to check that $C_{1,1}^2$ is not self-orthogonal, so we suppose $n > 1$.  In particular, $k = \frac{n(q-1)+1}{2} < n(q-1)$. Theorem \ref{thm-Sor} implies that ${C_{n,k}^q}^\perp = {\mbox{Span}}_{\F_q}  \{ {\bf 1}, C_{n,\ell}^q \}$. The argument in Remark \ref{rem-no1} implies that ${\bf 1}\not\in C_{n,\ell}^q$, so
\[
\dim({C_{n,k}^q}^\perp)  = \dim(C_{n,\ell}^q) +1 \le \dim(C_{n,k}^q).
\]
If $C$ is a linear code for which $\dim(C^\perp) \le \dim(C)$, then $C$ is self-orthogonal if and only if it is self-dual.  But, ${C_{n,k}^q}^\perp$ contains ${\bf 1}$, and the argument in Remark \ref{rem-no1} shows that $C_{n,k}^q$ does not.  We conclude that $C_{n,k}^q$ is not self-orthogonal.

\item Suppose $k - \ell \ge 2$.  By Theorem \ref{thm-Sor}, ${C_{n,k}^q}^\perp \subseteq {\mbox{Span}}_{\F_q}  \{ {\bf 1}, C_{n,\ell}^q \}$, which implies $\dim({C_{n,k}^q}^\perp) \le \dim(C_{n,\ell}^q) + 1$.  By Corollary \ref{cor-dim_diff}, $\dim(C_{n,k}^q) - \dim(C_{n,\ell}^q) \ge k-\ell \ge 2$.  We conclude that $\dim(C_{n,k}^q) > \dim({C_{n,k}^q}^\perp)$, so $C_{n,k}^q$ is not self-orthogonal.

\end{enumerate}

Now suppose that $1 \le k \le \frac{n(q-1)}{2}$ and that $2k \not\equiv 0 \pmod{q-1}$.  Let $G$ be a generator matrix for $C_{n,k}^q$.  We will show that $G G^T$ has a nonzero entry.  Applying Proposition \ref{prop-hull} then shows that $C_{n,k}^q$ is not self-orthogonal.

Choose a basis for $C_{n,k}^q$ including the monomial $x_n^k$.  Let $G$ be a generator matrix for $C_{n,k}^q$ for which the first row is equal to $\ev(x_n^k)$.  Throughout the rest of this proof, let $f(x_0,\ldots, x_n) = x_n^{2k}$.  The top left entry of $G G^T$ is equal to $\sum_{p \in \mathcal{P}'} f(p)$.

We recall the basic fact that for any positive integer $r$, we have
\begin{equation}\label{eq-sum}
\sum_{\beta\in \F_q} \beta^r =
\begin{cases}
-1 & \text{if } r \equiv 0 \pmod{q-1} \\
0 & \text{otherwise}
\end{cases}.
\end{equation}
Because of our choice of affine representatives $\mathcal{P}'$ and the fact that $2k \not\equiv 0 \pmod{q-1}$, we have $\sum_{p \in \mathcal{P}'} f(p) = f(0,\ldots,0,1) = 1^{2k} = 1$.  Therefore, $GG^T$ has a nonzero entry and $C_{n,k}^q$ is not self-orthogonal.
\end{proof}

Theorem~\ref{thm-so-1} directly implies the following result.
\begin{cor} \label{cor-hull-1}
Suppose that $1 \le k \le \frac{n(q-1)}{2}$ and $2k \equiv 0 \pmod{q-1}$.  Then $C^q_{n,k} \subseteq {C^q_{n,k}}^{\perp}$ and so $\Hull(C^q_{n,k}) = C^q_{n,k}$.
\end{cor}

\begin{example}{\rm
For any odd prime power $q$, $C_{2,\frac{q-1}{2}}^q$ is a self-orthogonal code with
parameters $[N=q^2+q+1, K=\frac{(q+3)(q+1)}{8}, D=\frac{(q+3)}{2} q]$. For example, when $q=3$, $C_{2,1}^q$ is a self-orthogonal $[13, 3, 9]$ code over $\mathbb F_3$, whose parameters are optimal based on Grassl's table~\cite{Gra}. When $q=5$, $C_{2,2}^q$ is a self-orthogonal $[31, 6, 20]$ code over $\mathbb F_5$, whose parameters are best known based on Grassl's table~\cite{Gra}.
}
\end{example}

The following corollary describes a situation for which ${C_{n,k}^q}^\perp$ is self-orthogonal.
\begin{cor}\label{cor-dual_so}
Suppose that $\frac{n(q-1)}{2} \le k \le n(q-1)$ and $2k \equiv 0 \pmod{q-1}$ but $k \not\equiv 0 \pmod{q-1}$.  Then ${C_{n,k}^q}^\perp$ is self-orthogonal.  That is, for any odd positive integer $n \le r < 2n$, we have ${C_{n,\frac{r(q-1)}{2}}^q}^\perp$ is self-orthogonal.
\end{cor}
\begin{proof}
Let $\ell = n(q-1) - k$ and note that $\ell \le \frac{n(q-1)}{2}$.  Since $2\ell \equiv -2k \pmod{q-1}$, we see that $2\ell \equiv 0 \pmod{q-1}$, but $\ell \equiv -k \not\equiv 0 \pmod{q-1}$.  Theorem \ref{thm-Sor} implies that ${C_{n,k}^q}^\perp = C_{n,\ell}^q$ and Theorem \ref{thm-so-1} implies that $C_{n,\ell}^q$ is self-orthogonal.
\end{proof}

\begin{example}{\rm
This corollary shows that ${C_{2,3}^3}^\perp$ is self-orthogonal since $3-1 \le 3  \le 2(3-1)$ and $6\equiv 0 \pmod{2}$, but $3 \not\equiv 0 \pmod{2}$.  This examples comes from choosing $r=3$ in the second part of the statement of the corollary.
}
\end{example}

We next characterize the projective Reed-Muller codes that are LCD. That is, we determine the $(n,k,q)$  for which $C^q_{n,k} \cap {C_{n,k}^q}^\perp$ is trivial.

\begin{theorem}\label{thm-LCD}
Let $1\le k < n(q-1)$.  Then $C^q_{n,k} \cap {C_{n,k}^q}^\perp$ is nontrivial.  That is, $C_{n,k}^q$ is not LCD.
\end{theorem}

\begin{remark}
As discussed at the end of Remark \ref{rem-no1}, $C_{n,n(q-1)}^q$ is LCD.  Taken together with Theorem \ref{thm-LCD} we see that if $1\le k \le n(q-1)$, then $C_{n,k}^q$ is LCD if and only if $k = n(q-1)$.
\end{remark}

\begin{proof}[Proof of Theorem \ref{thm-LCD}]
We prove that $\Hull(C)$ is nontrivial by showing that $\ev(x_0^k) \in C^q_{n,k} \cap {C_{n,k}^q}^\perp$.  Since $x_0^k$ is a monomial of degree $k$, it is clear that $\ev(x_0^k) \in C_{n,k}^q$.  We need only show that that $\ev(x_0^k)$ is orthogonal to $\ev(m)$ for every monomial $m = x_0^{b_0} x_1^{b_1} \cdots x_n^{b_n}$ where each $b_i$ satisfies $0 \le b_i \le k$ and $b_0 + \cdots + b_n = k$.

We compute that $\langle \ev(x_0^k),\ev(m)\rangle$ is equal to
\begin{eqnarray*}
\sum_{p \in \mathcal{P}'} (x_0^k m)(p) & = & \sum_{\alpha_1,\ldots, \alpha_{n} \in \F_q} 1^{k+b_0} \alpha_1^{b_1} \alpha_2^{b_2}\cdots \alpha_{n}^{b_{n}}   \\
& & + \sum_{\alpha_2,\ldots, \alpha_{n} \in \F_q} 0^{k+b_0} 1^{b_{1}}
\alpha_2^{b_2} \cdots \alpha_{n}^{b_{n}}  \\
& & + \cdots \\
& & + \sum_{\alpha_n\in \F_q}  \left(0^{k+b_0} \cdots 0^{b_{n-2}} \right)1^{b_{n-1}}  \alpha_n^{b_n} \\
& & + \left(0^{k+b_0}\cdots 0^{b_{n-1}}\right) 1^{b_n}\\
 & = &  \sum_{\alpha_1,\ldots, \alpha_{n} \in \F_q} \alpha_1^{b_1} \alpha_2^{b_2}\cdots \alpha_{n}^{b_{n}} \\
  & = &  \left(\sum_{\alpha_1\in \F_q} \alpha_1^{b_1}\right)  \left(\sum_{\alpha_2\in \F_q} \alpha_2^{b_2}\right) \cdots  \left(\sum_{\alpha_n\in \F_q} \alpha_n^{b_n}\right).
\end{eqnarray*}
For a nonnegative integer $b$, Equation \eqref{eq-sum} implies that $\sum_{\alpha \in \F_q} \alpha^b$ is $0$ unless $b \equiv 0 \pmod{q-1}$ and $b$ is positive.  Since $b_1+\cdots + b_n \le k < n(q-1)$, it is not possible for each of the sums $\sum_{\alpha_i\in \F_q} \alpha_i^{b_i}$ to simultaneously be nonzero.  We conclude that $\ev(x_0^k)$ is orthogonal to $\ev(m)$.  Since $m$ was an arbitrary monomial of degree $k$, we see that $\ev(x_0^k) \in {C_{n,k}^q}^\perp$.
\end{proof}

\section{Hulls of Projective Reed-Muller Codes}\label{sec:main2}

In the previous section we determined when projective Reed-Muller codes have certain special properties, that is, when they are self-dual, when they are self-orthogonal, and when they are LCD.  For most parameters, the projective Reed-Muller code $C_{n,k}^q$ does not satisfy any of these special conditions, but we would still like to understand the dimension of its hull.  In this section we will determine the dimension of $\Hull(C_{n,k}^q)$ when $k < q-1$.  Replacing $C_{n,k}^q$ with its dual and noting that for any code $C$, $\Hull(C) = \Hull(C^\perp)$, we will also determine the dimension of $\Hull(C_{n,k}^q)$ when $(n-1)(q-1) < k < n(q-1)$.  We then consider the case where $n = 2$ and see how this recovers a result of Ruano and San-Jos\'e about dimensions of hulls of projective Reed-Muller codes over the projective plane.  At the end of this section we summarize the range of parameters for which our results determine the dimension of $\Hull(C_{n,k}^q)$.

We begin by determining $\dim(\Hull(C_{n,k}^q))$ when $k$ is not too large relative to $q$.
\begin{theorem}\label{thm-q_large}
Suppose that $q > 2k+1$. Then
\[
\dim(\Hull(C^q_{n,k})) = \dim (C^q_{n,k})-1 =\binom{n+k}{k}-1.
\]
Moreover, a basis for $\Hull(C_{n,k}^q)$ is given by $\{\ev(m)\}_{m\in \mathcal{M}}$ where $\mathcal{M}$ is the set of monomials of degree $k$ in $x_0,x_1,\ldots, x_n$ except for the monomial $x_n^k$.
\end{theorem}

\begin{proof}
By Proposition \ref{prop-hull} it suffices to show that $\rank(G G^T) = 1$ where $G$ is a generator matrix for $C^q_{n,k}$.  Since $q > 2k+1$, Lemma \ref{lem-parameters} implies that $C^q_{n,k}$ has dimension $\binom{n+k}{k}$.  We take the lexicographic ordering of the $\binom{n+k}{k}$ monomials of degree $k$ in $x_0,\ldots, x_n$. We write this ordered list of monomials as $m_1,\ldots, m_{\binom{n+k}{k}}$.  Note that $m_{\binom{n+k}{k}} = x_n^k$.  It is clear that $\ev(m_1),\ldots, \ev(m_{\binom{n+k}{k}})$ form an ordered basis for $C_{n,k}^q$.  These choices determine a generator matrix $G$.  We will prove that $GG^T$ has a single nonzero entry, so in particular, $\rank(GG^T) = 1$.

The $(i,j)$-entry of $GG^T$ comes from taking the sum of the evaluations of the monomial $m_i  m_j$ over the points of $\mathcal{P}'$.  Suppose $m_i = x_0^{a_0} \cdots x_n^{a_n}$ and $m_j = x_0^{b_0} \cdots x_n^{b_n}$ where $0\le a_0,\ldots, a_n, b_0,\ldots, b_n \le k$ and $\sum_{\ell = 0}^n a_\ell = \sum_{\ell = 0}^n b_\ell = k$.  Recall that $0^0 = 1$ denotes the empty product, so if $r$ is a nonnegative integer, then $0^r = 1$ when $r=0$ and $0^r = 0$ otherwise.

Our choice of $\mathcal{P}'$ implies that the $(i,j)$-entry of $GG^T$ is
\begin{eqnarray*}
\sum_{p \in \mathcal{P}'} (m_i m_j)(p) & = & \sum_{\alpha_1,\ldots, \alpha_{n} \in \F_q} 1^{a_0+b_0} \alpha_1^{a_1+b_1} \alpha_2^{a_2+b_2}\cdots \alpha_{n}^{a_{n}+b_{n}}   \\
& & + \sum_{\alpha_2,\ldots, \alpha_{n} \in \F_q} 0^{a_0+b_0} 1^{a_{1}+b_{1}}
\alpha_2^{a_2+b_2} \cdots \alpha_{n}^{a_{n}+b_{n}}  \\
& & + \cdots \\
& & + \sum_{\alpha_n\in \F_q}  \left(0^{a_0+b_0} \cdots 0^{a_{n-2}+b_{n-2}} \right)1^{a_{n-1}+b_{n-1}}  \alpha_n^{a_n+b_n} \\
& & + \left(0^{a_0+b_0}\cdots 0^{a_{n-1}+b_{n-1}}\right) 1^{a_n+b_n}.
\end{eqnarray*}
For any $i$, we have
\begin{eqnarray*}
& &  \sum_{\alpha_{i+1},\ldots, \alpha_{n} \in \F_q} \left( 0^{a_{0}+b_{0}} \cdots 0^{a_{i-1}+b_{i-1}}\right) 1^{a_{i}+b_{i}}
\alpha_{i+1}^{a_{i+1}+b_{i+1}}\cdots \alpha_{n}^{a_{n}+b_{n}}   \\
 & = &  \left( 0^{a_{0}+b_{0}} \cdots 0^{a_{i-1}+b_{i-1}}\right) 1^{a_{i}+b_{i}}
  \left(\sum_{\alpha_{i+1}\in \F_q}  \alpha_{i+1}^{a_{i+1}+b_{i+1}}\right) \cdots \left(\sum_{\alpha_n\in \F_q}  \alpha_n^{a_n+b_n}\right).
\end{eqnarray*}
Recall from Equation \eqref{eq-sum} that for any nonnegative integer $r < q-1$, we have $\sum_{\beta\in \F_q} \beta^r =0$.  For any $i,\ a_i + b_i \le 2k$.  Since $2k< q-1$, we see that
\[
 \sum_{p \in \mathcal{P}'} (m_i m_j)(p) = \left(0^{a_0+b_0}\cdots 0^{a_{n-1}+b_{n-1}}\right) 1^{a_n+b_n}  =
 \begin{cases}
 1 & \text{ if } m_i = m_j = x_n^k \\
 0 & \text{ otherwise}
 \end{cases}.
\]
We conclude that $GG^T$ has a single nonzero entry that is equal to $1$, so in particular this matrix has rank $1$.  The claim about the basis for $\Hull(C_{n,k}^q)$ follows from the fact that every entry of $GG^T$ except for the one in the lower-right corner is $0$.
\end{proof}

\begin{remark}
The assumption $q > 2k+1$ plays an important role in the proof of Theorem \ref{thm-q_large} as $\sum_{\beta\in \F_q} \beta^{q-1} = -1\neq 0$.
\end{remark}

\begin{cor}\label{cor-q_large}
Suppose that $n(q-1)- \frac{q-1}{2} < k \le n(q-1)$ and let $\ell = n(q-1) - k$. Then 
\[
\dim(\Hull(C^q_{n,k}))=\dim (C^q_{n,\ell})-1 =\binom{n+\ell}{\ell}-1.
\]
\end{cor}
\begin{proof}
For a linear code $C$, it is clear that $\Hull(C) = \Hull(C^\perp)$.  The assumption on $k$ implies that either $k = n(q-1)$, or $k \not\equiv 0 \pmod{q-1}$.  In the first case, ${C_{n,n(q-1)}^q}^\perp$ is the $1$-dimensional code spanned by ${\bf 1}$.  In Remark \ref{rem-no1} we note that ${\bf 1} \not\in C_{n,n(q-1)}^q$, so the hull is trivial in this case.  In the second case, Theorem \ref{thm-Sor}(i) implies that ${C_{n,k}^q}^\perp = C_{n,\ell}^q$.  We see that $1 \le \ell < \frac{q-1}{2}$, which implies $q > 2\ell+1$. Applying Theorem \ref{thm-q_large} completes the proof.
\end{proof}

With additional effort, the ideas that go into the proof of Theorem \ref{thm-q_large} can be adapted to compute the dimension of the hull of $C_{n,k}^q$ for a wider range of parameters.

\begin{theorem}\label{thm-k_to_qm1}
Let $\frac{q-1}{2} < k < q-1$.  Then 
\[
\dim(\Hull(C_{n,k}^q)) = \dim(C_{n,k}^q) - (2k+1 -(q-1)).
\]
Moreover, a basis for $\Hull(C_{n,k}^q)$ is given by $\{\ev(m)\}_{m\in \mathcal{M}}$ where $\mathcal{M}$ is the set of monomials of degree $k$ in $x_0,x_1,\ldots, x_n$ except for the monomials $x_{n-1}^{k-a} x_n^a$ where $q-1-k \le a \le k$.
\end{theorem}

\begin{proof}
We follow the strategy from the proof of Theorem \ref{thm-q_large}.  Since $k < q$, Lemma \ref{lem-parameters} implies that $C^q_{n,k}$ has dimension $\binom{n+k}{k}$.  We take the lexicographic ordering of the $\binom{n+k}{k}$ monomials of degree $k$ in $x_0,\ldots, x_n$ and write this ordered list of monomials as $m_1,\ldots, m_{\binom{n+k}{k}}$.  It is clear that $\ev(m_1),\ldots, \ev(m_{\binom{n+k}{k}})$ form an ordered basis for $C_{n,k}^q$.  These choices determine a generator matrix $G$.  We will characterize the nonzero entries of $GG^T$. The form of this matrix will make it clear that it has rank $2k+1-(q-1)$, and moreover, will make it clear which monomials of degree $k$ correspond to codewords of $C_{n,k}^q$ that also lie in ${C_{n,k}^q}^\perp$.

Suppose $m_i = x_0^{a_0} \cdots x_n^{a_n}$ and $m_j = x_0^{b_0} \cdots x_n^{b_n}$ where $0\le a_0,\ldots, a_n, b_0,\ldots, b_n \le k$ and $\sum_{\ell = 0}^n a_\ell = \sum_{\ell = 0}^n b_\ell = k$. Following the reasoning from the proof of Theorem \ref{thm-q_large}, we see that the $(i,j)$-entry of $GG^T$ is given by
\[
\sum_{p \in \mathcal{P}'} (m_i m_j)(p)
=
\sum_{i=0}^n \left(
 \left( 0^{a_{0}+b_{0}} \cdots 0^{a_{i-1}+b_{i-1}}\right) 1^{a_{i}+b_{i}}
  \left(\sum_{\alpha_{i+1}\in \F_q}  \alpha_{i+1}^{a_{i+1}+b_{i+1}}\right) \cdots \left(\sum_{\alpha_n\in \F_q}  \alpha_n^{a_n+b_n}\right) \right).
\]
Equation \eqref{eq-sum} implies that $\sum_{\alpha \in \F_q} \alpha^{a+b} = 0$ unless $a+b$ is positive and $a+b \equiv 0 \pmod{q-1}$.  In particular, this is zero unless $a+b \ge q-1$.  Since $2k < 2(q-1)$, there is at most one $i$ for which $a_i + b_i$ is be both positive and divisible by $q-1$.  This implies that for any $i \le n-2$,
\[
 \left( 0^{a_{0}+b_{0}} \cdots 0^{a_{i-1}+b_{i-1}}\right) 1^{a_{i}+b_{i}}
  \left(\sum_{\alpha_{i+1}\in \F_q}  \alpha_{i+1}^{a_{i+1}+b_{i+1}}\right) \cdots \left(\sum_{\alpha_n\in \F_q}  \alpha_n^{a_n+b_n}\right) = 0.
\]
We see that
\[
\sum_{p \in \mathcal{P}'} (m_i m_j)(p)
=
\left(0^{a_{0}+b_{0}} \cdots 0^{a_{n-2}+b_{n-2}} \right)1^{a_{n-1}+b_{n-1}} \left(\sum_{\alpha_n\in \F_q}  \alpha_n^{a_n+b_n}\right)
+
0^{a_0+b_0} \cdots 0^{a_{n-1}+b_{n-1}} 1^{a_n+b_n}.
\]
As in the end of the proof of Theorem \ref{thm-q_large}, we have
\[
\left(0^{a_0+b_0}\cdots 0^{a_{n-1}+b_{n-1}}\right) 1^{a_n+b_n}  =
 \begin{cases}
 1 & \text{ if } m_i = m_j = x_n^k \\
 0 & \text{ otherwise}
 \end{cases}.
\]
We see that $GG^T$ has a $1$ in its lower-right corner, which corresponds to $m_i = m_j = x_n^k$.

We now assume that we are not in this special case, meaning that either $m_i$ or $m_j$ is not equal to $x_n^k$.  Therefore,
\[
\sum_{p \in \mathcal{P}'} (m_i m_j)(p)
=
\left(0^{a_{0}+b_{0}} \cdots 0^{a_{n-2}+b_{n-2}} \right) 1^{a_{n-1}+b_{n-1}} \left(\sum_{\alpha_n\in \F_q}  \alpha_n^{a_n+b_n}\right).
\]
This expression is zero unless $a_0,b_0, a_1, b_1,\ldots, a_{n-2}, b_{n-2} = 0$.  In order for this expression to be nonzero it is also necessary that $a_n + b_n = q-1$. Suppose $m_i = x_{n-1}^{k-a} x_n^{a}$ where $0 \le a \le k$.  Then $\sum_{p \in \mathcal{P}'} (m_i m_j)(p) = 0$ if $m_j$ is not equal to $x_{n-1}^{k-(q-1-a)} x_n^{q-1-a}$.  This is a monomial of degree $k$ if and only if $k - (q-1-a) \ge 0$, which is equivalent to $a \ge q-1-k$.

Each integer $a$ satisfying $(q-1)-k \le a \le k$ gives a nonzero entry of $GG^T$.  Together with the entry in the bottom-right corner of $GG^T$, this gives $2k-(q-1)+1 + 1$ nonzero entries of this matrix.  Because we used lexicographic order for the monomials of degree $k$ when defining $G$, the nonzero entries of $GG^T$ are contained in a submatrix in the lower-right corner of $GG^T$.  In this submatrix there are nonzero entries on the main antidiagonal, and then a single $1$ in the lower-right corner.  Since the last row and column of this matrix also contain nonzero entries corresponding to $\{m_i, m_j\} = \{x_n^k, x_{n-1}^{k-(q-1-k)} x_n^{q-1-k}\}$, we conclude that the rank of $GG^T$ is $2k-(q-1)+1$.

Every monomial $m$ that is not of the form $x_{n-1}^{k-a} x_n^a$ where $(q-1)-k \le a \le k$ has the property that the entire row corresponding to $m$ in $GG^T$ is zero.  The number of such monomials is the dimension of $\Hull(C_{n,k}^q)$, so these monomials form a basis for $\Hull(C_{n,k}^q)$.
\end{proof}

The proof of the following corollary follows the same strategy we used to deduce Corollary \ref{cor-q_large} from Theorem \ref{thm-q_large}.
\begin{cor}\label{cor-k_to_qm1}
Suppose that $(n-1)(q-1) < k < n(q-1)- \frac{q-1}{2}$ and let $\ell = n(q-1) - k$. Then 
\[
\dim(\Hull(C^q_{n,k})) = \dim (C^q_{n,\ell})- (2\ell+1 -(q-1)) =\binom{n+\ell}{\ell}-(2\ell+1 -(q-1)).
\]
\end{cor}

\begin{proof}
By the assumption on $k$, Theorem \ref{thm-Sor}(i) implies that ${C_{n,k}^q}^\perp = C_{n,\ell}^q$.  We see that $\frac{q-1}{2} < \ell < q-1$.  Applying Theorem \ref{thm-k_to_qm1} completes the proof.
\end{proof}

We now describe how the results of this section give a different proof of Theorem \ref{RSJ312} of Ruano and San-Jos\'e.  First note that if $k \in \left\{\frac{q-1}{2}, q-1, \frac{3(q-1)}{2}\right\}$, then Theorem \ref{thm-so-1} and Corollary \ref{cor-dual_so} determine $\Hull(C_{2,k}^q)$.

Now suppose that $1 \le k \le 2(q-1)$ and $2k \not\equiv 0 \pmod{q-1}$. Theorems \ref{thm-q_large} and \ref{thm-k_to_qm1} together with Corollaries \ref{cor-q_large} and \ref{cor-k_to_qm1} determine $\dim(\Hull(C_{2,k}^q))$. When $\min\{k,q-k-1\} = k$, the formula at the end of Theorem \ref{RSJ312} is equal to the formula in Theorem \ref{thm-q_large}, $\binom{k+2}{2} - 1$.  When $\min\{k,q-k-1\} = q-1-k$, since
\[
\binom{k+1}{2} + q-1-k = \binom{k+2}{2} - (2k+1-(q-1)),
\]
the formula from Theorem \ref{RSJ312} is equal to the formula in Theorem \ref{thm-k_to_qm1}.

\begin{itemize}
\item In the case where $k < \frac{q-1}{2}$, we see how the basis described in Theorem \ref{RSJ312} consists of all monomials in $x_0, x_1, x_2$ of degree $k$ except $x_2^k$, matching the result of Theorem \ref{thm-q_large}.
\item In the case where $k$ satisfies $\frac{q-1}{2} < k < q-1$, we see that the basis described in Theorem \ref{RSJ312} consists of all monomials in $x_0, x_1, x_2$ of degree $k$ except those of the form $x_1^{k-a} x_2^{a}$ where $a \in \{q-1-k,q-k,\ldots, k\}$, which matches the set of monomials described in Theorem \ref{thm-k_to_qm1}.
\end{itemize}
This completes our proof of Theorem \ref{RSJ312}.

Ruano and San-Jos\'e not only describe a basis for $\Hull(C_{2,k}^q)$ for each $q$, but more generally, they describe a basis for $C_{2,k_1}^q \cap C_{2,k_2}^q$ where $1\le k_1< k_2 \le 2(q-1)$ \cite[Theorem 3.9]{RSJ1}.  Our results in this paper for $n=2$ do not recover this theorem because we focus only on $C_{2,k}^q \cap {C_{2,k}^q}^\perp$.  We briefly describe how one could try to prove this more general result using our techniques.  In place of Proposition \ref{prop-hull}, we would want to apply \cite[Theorem 2.1]{GGJT}.  This would allow us to determine a basis for $C_{2,k_1}^q \cap {C_{2,k_2}^q}^\perp$ in terms of the nonzero entries of $G_1 G_2^T$, where $G_1$ is a generator matrix for $C_{2,k_1}^q$ and $G_2$ is a generator matrix for $C_{2,k_2}^q$.  The entries of this matrix correspond to summing a monomial of degree $k_1+k_2$ over a collection of affine representatives of the points of $\PP^2(\F_q)$.  The difficult case would be when both $k_1$ and $k_2$ are large, for example, when $q\le k_1$.  This does match up with \cite[Theorem 3.9]{RSJ1} where there is a special case that did not arise in our analysis of $C_{2,k}^q \cap {C_{2,k}^q}^\perp$.  We do not attempt to carry out this argument here.

We close this section by summarizing the cases where the dimension of $\Hull(C_{n,k}^q)$ is not determined by the results of this paper.
\begin{itemize}
\item When $1 \le k < \frac{q-1}{2}$, Theorem \ref{thm-q_large} determines the dimension of $\Hull(C_{n,k}^q)$.

\item When $n(q-1) - \frac{q-1}{2} < k \le n(q-1)$, Corollary \ref{cor-q_large} determines the dimension of $\Hull(C_{n,k}^q)$.

\item When $\frac{q-1}{2} < k < q-1$, Theorem \ref{thm-k_to_qm1} determines the dimension of $\Hull(C_{n,k}^q)$.

\item When $(n-1)(q-1) < k < n(q-1) - \frac{q-1}{2}$, Corollary \ref{cor-k_to_qm1} determines the dimension of $\Hull(C_{n,k}^q)$.

\item When $1\le k \le \frac{n(q-1)}{2}$ and $2k \equiv 0 \pmod{q-1}$, Theorem \ref{thm-so-1} determines the dimension of $\Hull(C_{n,k}^q)$.

\item When $\frac{n(q-1)}{2} \le k \le n(q-1)$ and $2k \equiv 0 \pmod{q-1}$ but $k \not\equiv 0 \pmod{q-1}$, Corollary \ref{cor-dual_so} determines the dimension of $\Hull(C_{n,k}^q)$.

\end{itemize}
In all other cases with $1\le k < n(q-1)$, Theorem \ref{thm-so-1} implies that $C_{n,k}^q$ is not self-orthogonal and Theorem \ref{thm-LCD} implies that $C_{n,k}^q$ is not LCD.  Therefore, $1\le \dim(\Hull(C_{n,k}^q)) < \dim(C_{n,k}^q)$, but we do not know the precise value of $\dim(\Hull(C_{n,k}^q))$.

A natural next case to consider is when $k$ satisfies $q-1 < k < \frac{3(q-1)}{2}$.  It seems likely that one could carry out an analysis similar to the one given in the proof of Theorem \ref{thm-k_to_qm1}, but the details would be more complicated.  This is the kind of thing that would be necessary in order to prove an analogue of Theorem \ref{RSJ312} for projective Reed-Muller codes $C_{3,k}^q$ corresponding to surfaces in $\PP^3$.  Investigating this further could be an interesting problem for future work.

\subsection*{Acknowledgments}
We thank Diego Ruano and Rodrigo San-Jos\'e for helpful comments.  N. Kaplan was supported by NSF Grant DMS 2154223. J.-L. Kim was supported by the National Research Foundation of Korea (NRF) Grant funded by the Korean government (NRF-2019R1A2C1088676).

\end{document}